\documentclass[conference]{IEEEtran}
\IEEEoverridecommandlockouts
\usepackage{cite}
\usepackage{amsmath,amssymb,amsfonts}
\usepackage{algorithmic}
\usepackage{graphicx}
\usepackage{glossaries}
\usepackage{textcomp}
\usepackage{amssymb}
\usepackage{amsthm}
\usepackage{amsmath}
\usepackage{subcaption}
\usepackage{tikz}
\DeclareMathOperator*{\argmax}{arg\,max}

\usepackage[linesnumbered,ruled,vlined]{algorithm2e}
\usepackage{algorithmic}
\usepackage{xcolor}
\usepackage{bm}

\newcommand{\E}{\mathbb{E}}
\newacronym{iot}{IoT}{Internet of Things}
\newacronym{ml}{ML}{machine learning}
\newacronym{dl}{DL}{Deep Learning}
\newacronym{marl}{MARL}{multi-agent reinforcement learning}
\newacronym{rl}{RL}{reinforcement learning}
\newacronym{decpomdp}{Dec-POMDP}{Decentralized Partially Observable Markov Decision Process }
\newacronym{pomdp}{POMDP}{Partially Observable Markov Decision Process}
\newacronym{uav}{UAV}{Unmanned Aerial Vehicle}
\newacronym{dqn}{DQN}{Deep Q-Network}
\newacronym{dnn}{DNN}{Deep Neural Network}
\newacronym{dial}{DIAL}{Differentiable Inter-Agent Learning}
\newacronym{mdp}{MDP}{Markov decision process}
\newacronym{fov}{FoV}{Field of View}
\newacronym{cnn}{CNN}{Convolutional Neural Network}
\newacronym{nn}{NN}{Neural Network}
\newacronym{ddql}{DDQL}{Distributed Deep Q-Learning}
\newacronym{pdf}{PDF}{Probability Density Function}
\newacronym{ndpomdp}{ND-POMDP}{Networked Distributed Partially Observable Markov Decision Process}
\newacronym{radam}{RAdam}{Rectified Adam}
\newacronym{cdf}{CDF}{cumulative distribution function}
\newacronym{mpc}{MPC}{Model Predictive Control}
\newacronym{rv}{rv}{Random Variable}
\newacronym{qoe}{QoE}{Quality of Experience}
\newacronym{tlc}{TLC}{Telecommunications}
\newacronym{cml}{CML}{communications for machine learning}
\newacronym{mlc}{MLC}{machine learning for communications}
\newacronym{drl}{DRL}{deep reinforcement rearning}
\newacronym{rf}{RF}{Radio Frequency}
\newacronym{urllc}{URLLC}{Ultra-Reliable and Low-Latency Communications}
\newacronym{fl}{FL}{federated learning}
\newacronym{kpi}{KPI}{Key Performance Indicators}
\newacronym{mec}{MEC}{Mobile Edge Computing}
\newacronym{ei}{EI}{Edge Intelligence}
\newacronym{bs}{BS}{base station}
\newacronym{sdn}{SDN}{Software Defined Networking}
\newacronym{mimo}{MIMO}{Multiple-Input Multiple-Output}
\newacronym{gp}{GP}{Gaussian Process}
\newacronym{iiot}{IIoT}{Industrial Internet of Things}
\newacronym{csi}{CSI}{Channel State Information}
\newacronym{sgd}{SGD}{Stochastic Gradient Descent}
\newacronym{iid}{i.i.d.}{independent and identically distributed}
\newacronym{ofdm}{OFDM}{Orthogonal Frequency Division Multiplexing}
\newacronym{los}{LOS}{Line-of-Sight}
\newacronym{nlos}{NLOS}{Non-Line-of-Sight}
\newacronym{snr}{SNR}{Signal to Noise Ratio}
\newacronym{rb}{RB}{Resource Block}
\newacronym{6g}{6G}{sixth generation}
\newacronym{ai}{AI}{artificial intelligence}
\newacronym{sfl}{SFL}{Synchronous Federated Learning}
\newacronym{frfl}{FRFL}{Fixed Rate Federated Learning}
\newacronym{pgm}{PGM}{Probabilistic Graphical Model}
\newacronym{hmm}{HMM}{Hidden Markov Model}
\newacronym{elbo}{ELBO}{Evidence Lower Bound}
\newacronym{pmf}{PMF}{Probability Mass Function}
\newacronym{smab}{SMAB}{Stochastic Multi-Armed Bandit}
\newacronym{mab}{MAB}{multi-armed bandit}
\newacronym{mc}{MC}{Monte Carlo}
\newacronym{is}{IS}{Importance Sampling}
\newacronym{dms}{DMS}{discrete memoryless source}
\newacronym{ucb}{UCB}{upper confidence bound}
\newacronym{ser}{SER}{Symbol Error Rate}
\newacronym{sc}{SC}{Semantic Communications}
\newacronym{voi}{VoI}{Value of Information}
\newacronym{nlp}{NLP}{natural language processing}
\newacronym{ts}{TS}{Thompson sampling}
\newacronym{cmab}{CMAB}{contextual multi-armed bandit}
\newacronym{rccmab}{RC-CMAB}{Rate-Constrained \gls{cmab}}
\newacronym{rcmab}{R-CMAB}{remote \gls{cmab}}
\newtheorem{theorem}{Theorem}[section]
\newtheorem{lemma}[theorem]{Lemma}

\IEEEoverridecommandlockouts
\newcommand\copyrightnotice{%
	\begin{tikzpicture}[remember picture,overlay]
	\node[anchor=north,yshift=-15pt] at (current page.north) {\parbox{\dimexpr\textwidth-\fboxsep-\fboxrule\relax}{
			\centering\footnotesize This paper has been submitted to IEEE ISIT 2022. Copyright may change without notice.}};
	\end{tikzpicture}
}

\def\BibTeX{{\rm B\kern-.05em{\sc i\kern-.025em b}\kern-.08em
    T\kern-.1667em\lower.7ex\hbox{E}\kern-.125emX}}
\begin{document}

\title{Remote Contextual Bandits}
 
\author{\IEEEauthorblockN{Francesco Pase}
\IEEEauthorblockA{\textit{University of Padova}\\
Padova, Italy \\
pasefrance@dei.unipd.it}
\and
\IEEEauthorblockN{Deniz G{\"u}nd{\"u}z}
\IEEEauthorblockA{\textit{Impreial College London} \\
London, UK \\
d.gunduz@imperial.ac.uk}
\and
\IEEEauthorblockN{Michele Zorzi}
\IEEEauthorblockA{\textit{University of Padova} \\
Padova, Italy \\
zorzi@dei.unipd.it}
}

\maketitle
\copyrightnotice
\begin{abstract}
	We consider a remote \gls{cmab} problem, in which the decision-maker observes the context and the reward, but must communicate the actions to be taken by the agents over a rate-limited communication channel. This can model, for example, a personalized ad placement application, where the content owner observes the individual visitors to its website, and hence has the context information, but must convey the ads that must be shown to each visitor to a separate entity that manages  the marketing content. In this \gls{rcmab} problem, the constraint on the communication rate between the decision-maker and the agents imposes a trade-off between the number of bits sent per agent and the acquired average reward. We are particularly interested in characterizing the rate required to achieve sub-linear regret. Consequently, this can be considered as a policy compression problem, where the distortion metric is induced by the learning objectives. We first study the fundamental information theoretic limits of this problem by letting the number of agents go to infinity, and study the regret achieved when Thompson sampling strategy is adopted. In particular, we identify two distinct rate regions resulting in linear and sub-linear regret behavior, respectively. Then, we provide upper bounds on the achievable regret when the decision-maker can reliably transmit the policy without distortion.
\end{abstract}

\section{Introduction}
\label{sec:intro}

In the last few years, synergies between \gls{ml} and communication networks have attracted a lot of interest in the research community, thanks to the fruitful interplay of the two fields in emerging applications, from Internet of Things (IoT) to autonomous vehicles, and other edge services. 
In most of these applications, both the generated data and the processing power are distributed across a network of physically distant devices, thus a reliable communication infrastructure is pivotal to run \gls{ml} algorithms that can leverage the collected distributed knowledge \cite{Park2019_edge, comm_to_learn_gunduz}. 
To this end, a lot of recent works have tried to redesign networks and to efficiently represent information to support distributed \gls{ml} applications, where the activities of data collection, processing, learning and inference are performed in different geographical locations; and therefore, the corresponding learning algorithms must take into account limited communication, memory, and processing resources, as well as addressing privacy issues. 

In contrast to the insatiable growth in our desire to gather more data and intelligence, available communication resources (bandwidth and power, in particular) are highly limited, and must be shared among many different devices and applications. This requires the design of highly communication-efficient distributed learning algorithms, particularly for edge applications. Information theory, and in particular rate-distortion theory, have laid the fundamental limits of efficient data compression, with the aim to reconstruct the source signal with the highest fidelity \cite{cover:IT}. However, in the aforementioned applications, the goal is often not to reconstruct the source signal, but to make inferences based on it. This requires \textit{task-oriented compression}, filtering out the unnecessary information for the target application, and thus decreasing the number of bits that have to be transmitted over the communication channels. This approach should target the questions of \textit{what} is the most useful information that has to be sent, and \textit{how} to represent it, in order to meet the application requirements consuming the minimum amount of network resources~\cite{Jankowski2021, Tung_rl_gunduz}. 

\begin{figure}[t!]
	\centering
	\includegraphics[width=0.97\linewidth]{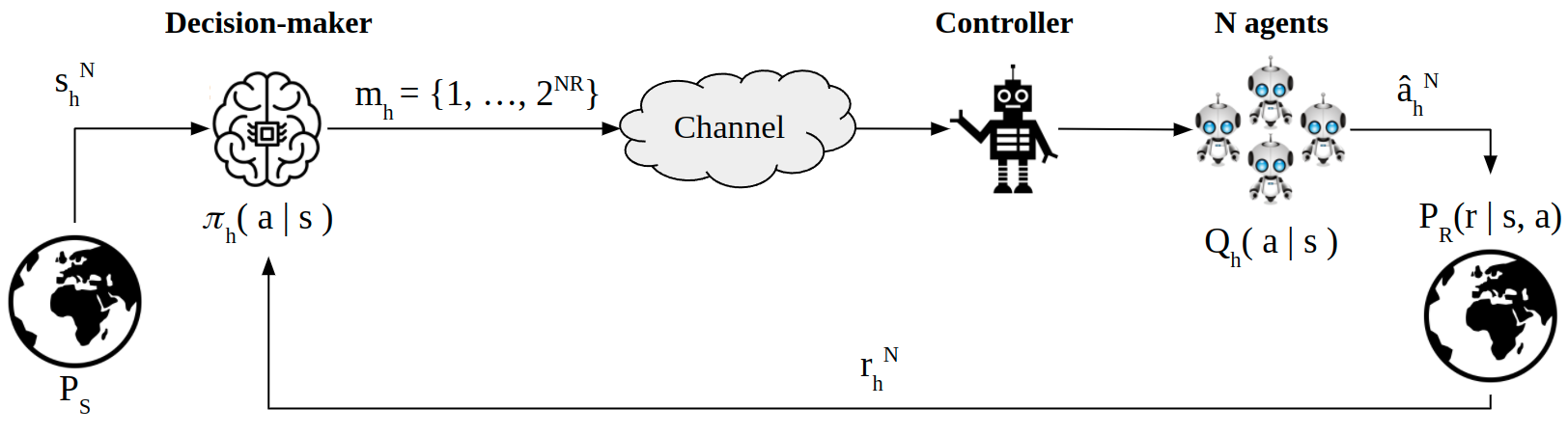}
	\caption{The \gls{rcmab} problem with a rate-limited communication channel.}
	\label{fig:tx_scheme}
\end{figure}

Our goal in this paper is to theoretically investigate a \glsfirst{cmab} problem, in which the context information is available to a remote \textit{decision-maker}, whereas the actions are taken by a remote entity, called the \textit{controller}, controlling a multitude of agents, each with an independent context. We can assume that a limited communication link is available between the decision-maker and the controller at each round to communicate the intended actions. The controller must decide on the action to be taken by each agent based on the message received over the channel, while the decision-maker observes the rewards at each round, and updates its policy accordingly. This framework is described in Fig.~\ref{fig:tx_scheme}.

This scenario can model, for example, a personalized ad placement application, where the content owner observes the individual visitors to its website, and hence has the context information, but must convey the ads that must be shown to each visitor to a separate entity that manages the marketing content. This will require communicating hundreds or thousands of ads to be placed at each round, chosen from a large set of possible ads, within the resource and delay constraints of the underlying communication channel, which is quantified as the number of bits available per agent. This problem may arise in other similar applications of \glspl{cmab} with communication constraints between the decision-maker and the controller \cite{bouneffouf2019survey}.


\section{Related Work}
\label{sec:related_work}
Given the amount of data that is generated by machines, sensors and mobile devices, the design of distributed learning algorithms is a hot topic in the \gls{ml} literature. These algorithms often impose communication constraints among agents, requiring the design of methods to allow efficient representation of messages to be exchanged. While rate-distortion theory deals with efficient lossy transmission of signals \cite{cover:IT}, in  \gls{ml} applications we typically do not need to reconstruct the underlying signal, but wish to make some inference based on it. These applications can be modeled through distributed hypothesis testing \cite{Berger_1979, Ahlswede-Csiszar, skreekumar:tit:2020} and estimation \cite{Zhang:NIPS:13, Xu:IT:17} problems under rate constraints. 


There is a growing literature on multi-agent \gls{rl} with communication links \cite{foerster_learning_2016, sukhbaatar_learning_2016, havrylov:nips:2017, lazaridou_multi-agent_2017}. These papers consider a multi-agent partially observable Markov decision process (POMDP), where the agents collaborate to resolve a specific task. In addition to the usual reward signals, agents can also benefit from the available communication links to better cooperate and coordinate their actions. It is shown that communication can help overcome the inherent non-stationarity of the multi-agent environment. Our problem can be considered as a special case of this general \gls{rl} formulation, where the state (context) at each time is independent of the past states and actions. Moreover, we focus on a particular setting in which the communication is one-way, from the decision-maker that observes the context and the reward, towards the controller that takes the actions. This formulation is different from the existing results in the literature involving multi-agent \gls{mab}. In \cite{agarwal2021multiagent}, each agent can pull an arm and communicate with others. They do not consider the contextual case, and focus on a particular communication scheme, where each agent shares the index of the best arm according to its own experience. Another related formulation is proposed in~\cite{fragouli}, where a pool of agents collaborate to solve a common \gls{mab} problem with a  rate-constrained communication channel from the agents to the server. In this case, agents observe their rewards and upload them to the server, which in turn updates the policy used to instruct them. In  \cite{park2021partial_ts}, the authors consider a partially observable \gls{cmab} scenario, where the agent has only partial information about the context. However, this paper does not consider any communication constraint, and the partial/ noisy view of the context is generated by nature. Differently from the existing literature, our goal is to identify the fundamental information theoretic limits of learning with communication constraints in this particular scenario.


\section{Problem Formulation}
\label{sec:problem_formulation}

\subsection{The Contextual Multi-Armed Bandit (CMAB) Problem}
\label{sub:mab}

We consider $N$ agents, which experience independent realizations of the same \gls{cmab} problem. The \gls{cmab} is a sequential decision game in which the environment imposes a probability distribution $P_S$ over a set of contexts, or states, $\mathcal{S}$, which is finite in our case. The game proceeds in rounds, and at each round $h = 1, \dots, H$, a realization of the state $s_h^n \in \mathcal{S}$ is sampled from distribution $P_S$ for each agent $n \in \mathcal{N} = \{1, \dots, N\}$, independently across time and agents. The decision-maker observes the states $\left\{s_h^n \right\}_{n=1}^N$, and chooses an action (or arm) $a_h^n \in \mathcal{A} = \{1, \dots, K\}$, for each agent, where $K$ is the total number of available actions, with probability $\pi_h(a_h^n | s_h^n)$. Once the actions have been taken, the environment returns rewards for all the agents following independent realizations of the same reward process, $r_h^n = r(s_h^n, a_h^n) \sim P_R(r|s_h^n, a_h^n), \quad \forall n \in \mathcal{N}$, which depends on the state and the action of the corresponding agent. The policy $\pi_h(a_h^n | s_h^n)$ used to sample the actions is a mapping $\pi_h : \mathcal{H}^{h-1} \times \mathcal{S} \rightarrow \Delta_K $. The set $\mathcal{H}^{h-1}$ contains all possible observations of the decision-maker, and $H(h-1) \in \mathcal{H}^{h-1}$ represents the knowledge accumulated by all the agents up to round $h-1$, i.e., ${H(h-1)}$ $ = \left\{ \left\{ \left(s_1^n, a_1^n, r_1^n \right) \right\}_{n=1}^N, \ldots, \left\{ \left(s_{h-1}^n, a_{h-1}^n, r_{h-1}^n \right) \right\}_{n=1}^N\right\} \in \mathcal{H}^{(h-1)}$. The set $\Delta_K$ is the K-dimensional simplex, containing all possible distributions over the set of actions. Based on the history of rewards up to round $h-1$, the decision-maker can optimize its policy to minimize the Bayesian system regret, that is defined as
\begin{equation}
\label{eq:bayesian_system_regret}
	\text{BR}(\pi, H) = \mathbb{E} \left[  \sum_{h=1}^H \sum_{n \in \mathcal{N}} \mu(s_h^n, a^*(s_h^n)) - \mu(s_h^n, A_h^n) \right],
\end{equation}
where $A_h^n$ is the action taken by agent $n$ at round $h$ using policy $\pi_h(a| s)$, which does not depend on $n$, i.e., the decision-maker adopts the same policy for all the agents, $\mu(s, a) = \mathbb{E} \left[ r(s, a)\right]$ is the average reward of action $a$ in state $s$, and ${a^*(s) = \argmax_{a \in \mathcal{A}} \mu(s, a)}$ is the optimal action for state $s$, i.e., the action with the highest expected reward, which is unknown at the beginning. The expectation is taken with respect to the state, action, and problem instance distributions.

\subsection{Remote CMAB }
\label{sub:rccmab}
In our scenario, the process of observing the system states is spatially separated from the process of taking actions. The environment states, $\left\{s_h^n \right\}_{n=1}^N$, are observed by a central entity, i.e., the decision-maker, that has  to communicate to the controller over a rate-limited communication channel, at each round $h$, the information about the actions $\left\{a_h^n \right\}_{n=1}^N$ the agents should take. 
Consequently, the problem is to communicate the action distribution, i.e., the policy $\pi_h(a|s)$, which depends on the specific state realizations, to the controller within the available communication resources.

Specifically, the decision-maker employs a function $f^{(N)}_h: \mathcal{H}^{h-1} \times \mathcal{S}^N \rightarrow \{1, 2, \ldots, B\}$ to map the observed history and the $N$ states at time $h$ to a message index to be transmitted over the channel. The controller, on the other hand, employs a function $g^{(N)}_h: \{1, 2, \ldots, B\} \rightarrow \mathcal{A}^N$ to map the received message to a set of actions for the agents. In general, both functions $f^{(N)}_h$ and $g^{(N)}_h$ can be stochastic. The Bayesian regret achieved by sequences $\left\{f^{(N)}_h, g^{(N)}_h\right\}_{h=1}^H$ is given by
	\begin{align}
	\begin{split}
	\label{eq:regret}
	\text{BR}&(H, \left(f, g\right)) =  \\
	& \mathbb{E} \left[ \sum_{h=1}^{H} \sum_{n \in \mathcal{N}} r(s_h^n, a^*(s_h^n) - r(s_h^n, g_h^n(m_h)) \right],
	\end{split}
	\end{align}
	where $ g_h^n(m_h)$ is the action taken by agent $n$ based on message $m_h = f^{(N)}_{h}\left(H(h-1), s_h^N\right)$ transmitted in round $h$, and here $s_h^N \in \mathcal{S}^N$ is the vector containing the states of all the agents.
We say that, for a given problem with $N$ agents, a rate $R$ is \textit{achievable} if there exist functions $\left\{f^{(N)}_h, g^{(N)}_h\right\}_{h=1}^H$ as defined above with rate $\frac1N \log_2 B \leq R$ and regret 
\begin{equation}
\label{eq:achievable_rate}
	\lim_{H \rightarrow \infty} \frac{\text{BR}\left( H, \left\{f_h^{(N)}, g_h^{(N)}\right\} \right)}{H} = 0,
\end{equation}
i.e., sub-linear in rounds.


If a rate $R \geq \log K$ is available, then the intended action for each agent can be easily conveyed to the controller, and so every policy $\pi_h$ that achieves sub-linear regret in the classical problem, can achieve the same regret in the remote version. However, in general, it may not be possible to convey the decision-maker's policy perfectly to the controller, and it is not clear whether distorted versions of the policy $\pi$ can obtain sub-linear regret. If this is the case, it would be possible to reduce the necessary communication rate, while still solving the underlying learning problem, by \textit{compressing} the policy $\pi$.

\section{Solution}
\label{sec:solution}

We first split the problem of learning a policy $\pi$ at the decision-maker, and of characterizing the required rate to convey it, when a fixed distortion between $\pi$ and the policy adopted by the agents $Q$ is allowed. We then study the problem exploiting \gls{ts}, which is a popular strategy to efficiently solve \gls{mab} problems, and characterize the required asymptotic rate to solve the problem. We also provide an upper bound on the Bayesian system regret when the \gls{ts} policy can be perfectly conveyed to the controller.

\subsection{The Asymptotic Policy Rate }
\label{sub:optimal_solution}

We model the environment  as a \gls{dms}, which generates at each round states from a finite alphabet $\mathcal{S}$ with probability $P_S$, emitting sequences of $N$ symbols $s^N = (s_1, \dots, s_N)$, one per agent. We then denote with $\hat{Q}_{s^N}(s)$
the empirical probability of state $s \in \mathcal{S}$ in $s^N$.
We also consider the sequence of actions $a^N$, and denote with $\hat{Q}_{z^N}(s,a)$ 
the empirical joint probability of the pair $(s, a)$ in $z^N=((s_1, a_1), \dots, (s_N, a_N))$. 
The whole picture can be seen in Fig.~\ref{fig:tx_scheme}, where the actions taken by the agents are denoted by $\hat{a}$ to indicate that they can differ from $a$ dictated by policy $\pi$.
We  assume that the distribution $P_S$ is known (or accurately estimated). 

The decision-maker can observe the realization $s^N$ of the contexts, and its task is to transmit an index $m \in \{1, \dots, B\}$ over the channel so that the controller can generate from $m$ the actions $a^N$, where $\hat{Q}_{s^Na^N}$ is as close to $P_{SA}(s, a) = P_S(s) \pi(a | s)$ as possible, where closeness depends on a distortion measure $\E [d(\hat{Q}_{S^NA^N}, P_{SA})]$, which in general is not an average of a per-letter distortion measure. The problem is a compression task in which the decision-maker has knowledge of the states $s^N$, and wants to transmit a conditional probability distribution $\pi_{A|S}$ to the agents, consuming the minimum amount of bits, in such a way that the empirical distribution $\hat{Q}_{s^Na^N}$ is close to the joint distribution $P_{SA}$ induced by the policy. 
For a distortion function $d(Q_{SA}, P_{SA})$ that is $1)$ nonnegative, $2)$ upper bounded by a constant $D_{max}$, $3)$ continuous in $Q_{SA}$, and $4)$ convex in $Q_{SA}$, in~\cite{CommDistribution} the authors provide the rate-distortion function $R(D)$, i.e., the minimum rate $R=\frac{\log_2 B}{N}$ bits per symbol such that $ \E [d(\hat{Q}_{S^NA^N}, P_{SA})] \leq D$, in the limit when $N$ is arbitrarily large. 
\begin{theorem}[\cite{CommDistribution}, Theorem 1]
    The rate-distortion function for the problem of communicating policies is
    \begin{align}
        \label{eq:rate_dist}
        R(D) = \min_{Q_{A|S} :  d(Q_{SA}, P_{SA}) \leq D} I(S; A)
	\end{align}
assuming the set of $Q_{S|A}$ satisfying $d(Q_{SA}, P_{SA}) \leq D$ is not empty.
\end{theorem}

Here $Q_{SA} = P_S Q_{A|S}$ is the joint probability induced by the environment distribution $P_S$ and policy $Q_{A|S}$, which depends on the information sent by the decision-maker. As we can see, in the asymptotic limit of $N$ agents, the problem admits a single-letter solution, which also serves as a lower bound on the finite agent scenario. When imposing $D=0$, the needed rate is the mutual information between the states and actions, which are related by the policy $\pi$. Furthermore, if we allow $D>0$, Eq.~(\ref{eq:rate_dist}) characterizes the minimum rate needed to convey the actions to the controller. However, finding a closed form solution for the rate-distortion function is not a trivial task in general.



\subsection{Thompson Sampling (TS)}
\label{sub:ts}

In the proposed solution, the decision-maker adopts the \gls{ts} strategy \cite{thompson} to learn a policy. The reason why \gls{ts} is adopted is because, among the state-of-the-art \gls{mab} solutions, it relies on posterior sampling~\cite{daniel_benji}, that can be exploited within one round to sample different actions in parallel across the $N$ agents. If \gls{ucb} style algorithms are used, it is not clear how to modify them to perform exploration within one round, given that the policy is deterministic, and it chooses the action that maximizes the upper bound. Consequently, the action probability distribution induced by \gls{ts} is exploited in the \gls{rcmab} problem to perform exploration in parallel, and to further compress the original policy using Eq.~(\ref{eq:rate_dist}).

In particular, the decision-maker implements one \gls{ts} instance for each state $s \in \mathcal{S}$. Indeed, in our general formulation, there is no known structure between the states and rewards to be exploited. Consequently, the decision-maker maintains estimates of the distributions $p^{s,a}_h(\mu)$ of the mean reward $\mu(s, a) \in \mathbb{R}$, $\forall s \in \mathcal{S}$, $\forall a \in \mathcal{A}$. To take a decision in state $s_h$, the decision-maker samples $\hat{\mu}_h(s_h, a)  \sim p_h^{s_h,a}$, $\forall a \in \mathcal{A}$, and takes the action $a^* = \argmax_{a \in \mathcal{A}} \{\hat{\mu}_h(s_h, a)\}$. This procedure is repeated for each agent $n \in \mathcal{N}$. After receiving the rewards $\left\{r_h^n\right\}_{n=1}^N$, the decision-maker can update its belief on $\mu(s, a)$, i.e., the probabilities $p^{s,a}_h(\mu)$, in order to minimize the regret. We notice that this strategy induces a probability distribution $\pi_h(a | s)$ over the actions that is $
\pi_h(a|s) = \int_{\mathcal{D}} p_h^{s,a}(\mu) \prod_{j=1, j \neq a}^{K} P_h^{s,j}(\mu) d\mu$, where $P_h^{s,j}(\mu)$ is the \gls{cdf} of $\mu(s,j)$, and the random variables $\mu(s, a)$ are considered independently distributed.

However, in our scenario, the constraint on the rate imposed by the communication channel can make it infeasible for the controller to sample the actions directly from the true distribution $\pi_h(a | s)$. The agents have to use a proxy $Q_h(a|s)$, which is the one obtained from the message received over the channel. This problem is similar to approximate \gls{ts}, where a proxy distribution is used to sample the actions, or the reward means, given that the true distribution is too complex to sample from. In that case, the bottleneck is due to the complexity of sampling from the true mean reward distribution, whereas in this work, it is imposed by the limited-rate communication channel between the decision-maker and the controller.

\subsection{Asymptotic Limit for the Achievable Rate}
\label{subsec:regret}

To prove the results on the achievable regret of the \gls{ts} strategy, we adopt Assumption $1$ in~\cite{russo16}, that considers rewards to be distributed following canonical exponential families, and the priors used by \gls{ts} to be bounded away from zero $\forall (s,a)$.

In the following, we provide the minimum  rate needed to achieve sub-linear regret in all the states, $s \in \mathcal{S}$, when the decision-maker adopts \gls{ts} to learn the optimal actions. We define $H(A^*)$ as the entropy of the optimal arm, which we assume unique, or uniquely determined within a set of optimal arms, and computed based on the marginal $\pi^*(a) = \sum_s P_S(s) \pi^*(a|s)$, where $\pi^*$ is the optimal policy, and we prove that it is the minimum rate required.

We will use the following result from \cite{ayfer2020}.
\begin{theorem} [\cite{ayfer2020}, Theorem 2] \label{thm:ayfer}
	Suppose that the \gls{ts} policy $\pi(a|s)$ achieves sub-linear regret in each state $s \in \mathcal{S}$, then
	\begin{equation}
		\lim_{h \rightarrow \infty }\pi_h(a^*(s)|s) = 1 \quad \text{a.s.}
	\end{equation}
	where 
	\begin{equation*}
		a^*(s) = \argmax_{a \in \mathcal{A}} \mu(s, a).
	\end{equation*}
\end{theorem}
We now provide the following lemma.
\begin{lemma}\label{lemma:limit_rate}
	Assuming that Thompson Sampling policy $\pi_h(a|s)$ achieves sub-linear Bayesian system regret, then 
	\begin{equation}
		\lim_{h \rightarrow \infty} \text{I}_{\pi_h}(S;A) = \lim_{h \rightarrow \infty} H_{\pi_h}(A) = H(A^*).
	\end{equation}
\end{lemma}
\begin{proof} [Sketch of the Proof]
	The proof follows from Theorem \ref{thm:ayfer}, whose consequence is that, in the limit, the entropy of the \gls{ts} policy conditioned on state $s$ is zero. By using this with the definition of the rate provided in Eq.~(\ref{eq:rate_dist}), it is possible to conclude the proof.
\end{proof}
Theorem~\ref{thm:ayfer} and Lemma~\ref{lemma:limit_rate} are useful to prove the following results. Here the available rate $R$ is considered fixed in all rounds $h = \{1, \ldots, H\}$.

\begin{lemma}\label{lemma:min_rate}
	If $R < H(A^*)$, then it is not possible to convey a policy $Q(a|s)$ that achieves sub-linear Bayesian system regret.
\end{lemma}
\begin{proof} [Sketch of the Proof]
	If $R < H(A^*)$, from Eq.~(\ref{eq:rate_dist}), the policy $Q$ conveyed to the controller will have non-zero distortion $d(Q_{SA}, \pi_{SA}) = D > 0$ from $\pi$, $\forall h \in \left\{ 1, \ldots, H \right\}$. If we take, for example, the total variation as the distortion measure, in each round $h$, $Q$ would sample a sub-optimal arm with constant probability of at least $D$. Consequently, a sub-linear regret cannot be achieved.
\end{proof}
The following Lemma provides the achievability part. 
\begin{lemma}
	\label{lemma:achivable_rate}
	If $R > H(A^*)$, then achieving sub-linear regret is possible in all states $s \in \mathcal{S}$, as $N \rightarrow \infty$.
\end{lemma}
\begin{proof} [Sketch of the Proof]
	The intuition is that, even though during training the required rate $R_h$ to convey the current policy may exceed $R$, exploration is never penalized (actually it is enforced by the system). Consequently, \gls{ts} will converge to the optimal policy \cite{approx_ts}, that can be eventually perfectly transmitted to the controller, given that $R > H(A^*)$, which is the rate required in the limit as $N \rightarrow \infty$. This, together with the fact that \gls{ts} achieves sub-linear regret in this parallel multi-agent version of the problem (Theorem~\ref{thm:sys_regret}), concludes the proof.
\end{proof}

The consequence of Lemma \ref{lemma:achivable_rate} is that, even if the exact \gls{ts} policy $\pi_h$ cannot be transmitted $\forall h$, as long as $R > H(A^*)$, it is still possible to achieve sub-linear regret. According to the definition in Eq.~(\ref{eq:achievable_rate}), this implies that, as $N \rightarrow \infty$, any rate $R > H(A^*)$ is achievable, while any rate $R<H(A^*)$ is not achievable.

\subsection{Regret of the Optimal Policy }
\label{subsec:regret_bound}

In this section, we present both finite-time and asymptotic upper bounds on the regret obtained by the \gls{ts} strategy, when the policy $\pi_h$ can be perfectly transmitted at each round $h$. We further provide the per-agent regret, defined as the one obtained by a single agent. However, to fairly compare the obtained regret with \gls{ts} applied to the standard \gls{cmab} problem, we write them as a function of the virtual time-steps $t \in \{1, \ldots, T\}$, with $T=NH$, i.e., it represents the total number of interactions the system has with the environment through the agents. Indeed, the problem is mathematically equivalent to a single-agent \gls{cmab}, in which the parallel interactions of the $N$ agents are unrolled in time, with the additional constraint that the policy $\pi$ can be updated only every $N$ time-steps, i.e., at time-steps $t = Nh$ for $h \in \{1, \ldots, H\}$.

\begin{theorem} [Bayesian System Regret]
	\label{thm:sys_regret}
	The Bayesian system regret of \gls{ts} is upper bounded by
	\begin{equation}
		\text{BR}(\pi, T) \leq 2K|S| N + 4\sqrt{\left(2 + 6\log T \right) K N |S| T},
	\end{equation}
	and the asymptotic regret is 
	\begin{equation}
		\text{BR}(\pi, T) \in \mathcal{O} \left( \sqrt{KT|S| \log T} \right).
	\end{equation}	
\end{theorem}
\begin{proof} [Sketch of the Proof]
	The proof follows similar arguments to those in \cite{daniel_benji}, Section 6, with the difference that during each round $h$, the policy adopted by the $N$ parallel agents is not sequentially optimized, but can be updated only at the end of the round. Consequently, a penalty of $\sqrt{N}$ appears on the upper bound of finite-time regret, as when $T$ is small, playing with a sub-optimal policy $N$ times in parallel amplifies the regret. The result follows from bounding the gap between the counter of the number of times a particular action has been sampled at time $t$, and the counter at the end of the previous round, which is the value used to update the policy and to construct the confidence bounds~\cite{daniel_benji}. In the asymptotic case, i.e., $T >> N$, this effect vanishes, as the gap is almost $N$.
\end{proof}
\begin{theorem} [Bayesian Agent Regret]
	\label{thm:agent_regret}
	The Bayesian per-agent regret of \gls{ts} is upper bounded by
	\begin{equation}
	\text{BR}(\pi, T) \leq 2K|S| + 4\sqrt{\frac{\left(2 + 6\log T \right) K |S| T}{N}},
	\end{equation}
	and the asymptotic regret is 
	\begin{equation}
	\text{BR}(\pi, T) \in \mathcal{O} \left(\frac1N \sqrt{KT|S| \log T}\right).
	\end{equation}
\end{theorem}
\begin{proof} [Sketch of the Proof]
	The proof relies on Theorem~\ref{thm:sys_regret}, and on the observation that the per-agent regret is equal to ${\text{BR}(\pi, H, n) = \frac{\text{BR}(\pi, H)}{N}}$, due to the symmetry of the problem. Indeed, each agent interacts with an \gls{iid} copy of the environment and, at each round $h$, adopts the policy $\pi_h(a|s)$ known by the decision-maker, and equal for all the agents $n \in \mathcal{N}$.
\end{proof}

\section{Numerical results }
\label{sec:numerical_results}

In this experiment we analyze the asymptotic rate required by the \gls{ts} policy to be conveyed, that serves as a lower bound for practical scnearios with finite $N$, in three different environments, representing different relations between the states and optimal actions. In all the scenarios, there are $16$ actions per state, and $16$ states that are sampled uniformly by the environment. The first scenario is called \textit{16 Groups}, and for each state $s_i \in \mathcal{S}$, the best average reward is given by arm $a_i$, $i \in \{0, \dots, 15\}$. In particular, the reward behind arm $a_j$ in state $s_i$ is a Bernoulli random variable with parameter $\mu(s_i, a_j) = 0.8$ if $i=j$, whereas $\mu(s_i, a_j) \sim \text{Unif}_{[0, 0.75]}$ if $i \neq j$. The best action is thus strongly correlated with the state, and a sufficiently high rate is required to sample from the optimal policy $\pi^*$. In the second experiment, the setting is similar to the one presented above, but the Bernoulli parameter $\mu(s_i, a_j)$ is $0.8$ if $\lfloor \frac{j}{2} \rfloor = i$, and sampled uniformly in $[0, 0.75]$ otherwise. Consequently, the best actions can be grouped into 8 different classes. This scenario is indicated as the \textit{8 Groups} experiment. The same procedure is applied to generate the last environment, except that the best responses are grouped into just $2$ different classes.

Fig.~\ref{fig:ratepolicy} shows the asymptotic rate needed to convey the \gls{ts} policy in the three described scenarios, as a function of the number of rounds. It is possible to observe that the policy rates are converging to $4, 3, 1$ bits, respectively, which are the mutual information values between the states and optimal actions, i.e., the entropies of uniform distributions over the different problem classes. We can also notice that, during the exploration phase at the beginning of the training process, very limited information has to be sent, whereas the required rate gradually increases as the decision-maker learns to map states to optimal responses.

\begin{figure}
	\centering
	\includegraphics[width=0.9\linewidth]{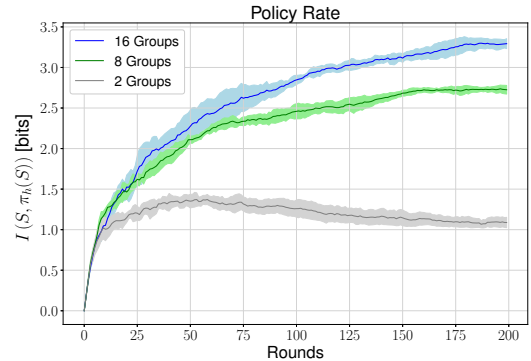}
	\caption{Asymptotic rate needed to reliably transmit the decision-maker's policy. Curves are average values $\pm 1 \sigma$, computed over $5$ independent runs per scenario.}
	\label{fig:ratepolicy}
\end{figure}

\section{Conclusion}
\label{sec:conclusion}

We have introduced and studied the \gls{rcmab} problem, in which an intelligent entity, i.e., the decision-maker, observes the contexts of $N$ parallel \gls{cmab} processes, and has to decide on the actions 
depending on the current contexts and the past actions and rewards. However, the actions are implemented by a controller that is connected to the decision-maker through a communication link. First, we cast the problem into the proper information-theoretic framework, and provided the needed rate to convey a policy, when admitting a maximum distortion between a compressed policy adopted by the controller and the one of the decision-maker. We then analyzed the problem when the \gls{ts} algorithm is used, and characterized the minimum achievable rate to obtain sub-linear regret. In the end, we provided finite-time and asymptotic upper bounds on the regret achieved by the system, when the policy can be conveyed to the controller. Ongoing work includes the formulation of the problem with specific distortion functions, which can be derived from the underlying learning objectives, and analysis of the behavior when non-zero distortion is allowed, or, equivalently, when the available rate is not sufficient to perfectly transmit the updated policy. 

\bibliography{biblio}

\begin{thebibliography}{10}
\providecommand{\url}[1]{#1}
\csname url@samestyle\endcsname
\providecommand{\newblock}{\relax}
\providecommand{\bibinfo}[2]{#2}
\providecommand{\BIBentrySTDinterwordspacing}{\spaceskip=0pt\relax}
\providecommand{\BIBentryALTinterwordstretchfactor}{4}
\providecommand{\BIBentryALTinterwordspacing}{\spaceskip=\fontdimen2\font plus
\BIBentryALTinterwordstretchfactor\fontdimen3\font minus
  \fontdimen4\font\relax}
\providecommand{\BIBforeignlanguage}[2]{{%
\expandafter\ifx\csname l@#1\endcsname\relax
\typeout{** WARNING: IEEEtran.bst: No hyphenation pattern has been}%
\typeout{** loaded for the language `#1'. Using the pattern for}%
\typeout{** the default language instead.}%
\else
\language=\csname l@#1\endcsname
\fi
#2}}
\providecommand{\BIBdecl}{\relax}
\BIBdecl

\bibitem{Park2019_edge}
J.~Park, S.~Samarakoon, M.~Bennis, and M.~Debbah, ``Wireless network
  intelligence at the edge,'' \emph{Proceedings of the IEEE}, vol. 107, no.~11,
  pp. 2204--2239, Nov. 2019.

\bibitem{comm_to_learn_gunduz}
D.~Gündüz, D.~B. Kurka, M.~Jankowski, M.~M. Amiri, E.~Ozfatura, and
  S.~Sreekumar, ``Communicate to learn at the edge,'' \emph{IEEE Communications
  Magazine}, vol.~58, no.~12, pp. 14--19, Jan. 2020.

\bibitem{cover:IT}
T.~M. Cover and J.~A. Thomas, \emph{Elements of Information Theory (Wiley
  Series in Telecommunications and Signal Processing)}.\hskip 1em plus 0.5em
  minus 0.4em\relax USA: Wiley-Interscience, 2006.

\bibitem{Jankowski2021}
M.~Jankowski, D.~Gündüz, and K.~Mikolajczyk, ``Wireless image retrieval at
  the edge,'' \emph{IEEE Journal on Selected Areas in Communications}, vol.~39,
  no.~1, pp. 89--100, Nov. 2021.

\bibitem{Tung_rl_gunduz}
T.-Y. Tung, S.~Kobus, J.~P. Roig, and D.~Gündüz, ``Effective communications:
  A joint learning and communication framework for multi-agent reinforcement
  learning over noisy channels,'' \emph{IEEE Journal on Selected Areas in
  Communications}, vol.~39, no.~8, pp. 2590--2603, Aug. 2021.

\bibitem{bouneffouf2019survey}
D.~Bouneffouf and I.~Rish, ``A survey on practical applications of multi-armed
  and contextual bandits,'' \emph{arXiv cs.LG:1904.10040}, 2019.

\bibitem{Berger_1979}
T.~Berger, ``Decentralized estimation and decision theory,'' in \emph{IEEE 7th.
  Spring Workshop on Inf. Theory}, Mt. Kisco, NY, Sep. 1979.

\bibitem{Ahlswede-Csiszar}
R.~Ahlswede and I.~Csisz\'{a}r, ``Hypothesis testing with communication
  constraints,'' \emph{IEEE Transactions on Information Theory}, vol.~32,
  no.~4, pp. 533--542, Jul. 1986.

\bibitem{skreekumar:tit:2020}
S.~Sreekumar and D.~Gündüz, ``Distributed hypothesis testing over discrete
  memoryless channels,'' \emph{IEEE Transactions on Information Theory},
  vol.~66, no.~4, Apr. 2020.

\bibitem{Zhang:NIPS:13}
Y.~Zhang, J.~Duchi, M.~I. Jordan, and M.~J. Wainwright, ``Information-theoretic
  lower bounds for distributed statistical estimation with communication
  constraints,'' in \emph{Advances in Neural Information Processing Systems},
  vol.~26, Dec. 2013.

\bibitem{Xu:IT:17}
A.~Xu and M.~Raginsky, ``Information-theoretic lower bounds on {B}ayes risk in
  decentralized estimation,'' \emph{IEEE Transactions on Information Theory},
  vol.~63, no.~3, pp. 1580--1600, Mar. 2017.

\bibitem{foerster_learning_2016}
J.~N. Foerster, Y.~M. Assael, N.~de~Freitas, and S.~Whiteson, ``Learning to
  {communicate} with {deep} {multi}-{agent} {reinforcement} {learning},''
  \emph{arXiv:1605.06676 [cs]}, May 2016, arXiv: 1605.06676.

\bibitem{sukhbaatar_learning_2016}
S.~Sukhbaatar, A.~Szlam, and R.~Fergus, ``Learning multiagent communication
  with backpropagation,'' in \emph{Proc. of 30th {Int'l} {Conf.} on {Neural}
  {Information} {Proc.} {Systems}}, ser. {NIPS}'16, Red Hook, NY, Dec. 2016,
  pp. 2252--2260.

\bibitem{havrylov:nips:2017}
S.~Havrylov and I.~Titov, ``Emergence of language with multi-agent games:
  Learning to communicate with sequences of symbols,'' in \emph{Advances in
  Neural Information Processing Systems}, Dec. 2017.

\bibitem{lazaridou_multi-agent_2017}
A.~Lazaridou, A.~Peysakhovich, and M.~Baroni, ``Multi-agent cooperation and the
  emergence of (natural) language,'' \emph{arXiv:1612.07182 [cs]}, Mar. 2017,
  arXiv: 1612.07182.

\bibitem{agarwal2021multiagent}
M.~Agarwal, V.~Aggarwal, and K.~Azizzadenesheli, ``Multi-agent multi-armed
  bandits with limited communication,'' in \emph{arXiv:2102.08462 [cs]}, 2021.

\bibitem{fragouli}
O.~A. Hanna, L.~F. Yang, and C.~Fragouli, ``Solving multi-arm bandit using a
  few bits of communication,'' in \emph{38th International Conference on
  Machine Learning}, 2021.

\bibitem{park2021partial_ts}
H.~Park and M.~K.~S. Faradonbeh, ``Analysis of {T}hompson sampling for
  partially observable contextual multi-armed bandits,'' \emph{arXiv:2110.12175
  [stat.ML]}, 2021.

\bibitem{CommDistribution}
G.~Kramer and S.~A. Savari, ``Communicating probability distributions,''
  \emph{IEEE Transactions on Information Theory}, vol.~53, no.~2, pp. 518--525,
  Feb. 2007.

\bibitem{thompson}
W.~R. Thompson, ``On the theory of apportionment,'' \emph{American Journal of
  Mathematics}, vol.~57, no.~2, pp. 450--456, Apr. 1935.

\bibitem{daniel_benji}
D.~Russo and B.~Van~Roy, ``Learning to optimize via posterior sampling,''
  \emph{Mathematics of Operation research}, vol.~39, no.~4, pp. 1221--1243,
  Nov. 2014.

\bibitem{russo16}
D.~Russo, ``Simple {Bayesian} algorithms for best arm identification,'' in
  \emph{29th Annual Conference on Learning Theory}, ser. Proceedings of Machine
  Learning Research, vol.~49.\hskip 1em plus 0.5em minus 0.4em\relax Columbia
  University, New York, New York, USA: PMLR, 23--26 Jun 2016, pp. 1417--1418.

\bibitem{ayfer2020}
C.~Kalkanli and A.~Ozgur, ``Asymptotic convergence of {T}hompson sampling,'' in
  \emph{arXiv:2011.03917v1}, 2020.

\bibitem{approx_ts}
M.~Phan, Y.~Abbasi~Yadkori, and J.~Domke, ``{T}hompson sampling with
  approximate inference,'' in \emph{Advances in Neural Information Processing
  Systems}, Dec. 2019.

\end{thebibliography}
\bibliographystyle{IEEEtran}
\end{document}